\documentclass[authoryear,5p,twocolumn]{elsarticle}
\journal{Automatica}
\usepackage{amssymb}
\usepackage{amsmath}
\usepackage{epsfig}
\usepackage{amsthm}
\usepackage{mathtools}
\usepackage[hidelinks]{hyperref}
\usepackage{dblfloatfix}
\hypersetup{
	colorlinks = true,
}

\newcommand{\norm}[1]{ \Vert #1 \Vert}

\newtheorem{thm}{Theorem}
\newtheorem{lem}{Lemma} 
\newdefinition{rmk}{Remark}
\newdefinition{defi}{Definition}
\newdefinition{assum}{Assumption}

\begin{document}
	
	\begin{frontmatter}		
	\title{Predicting the future state of disturbed LTI systems: A solution based on high-order observers\tnoteref{thanks}}
	
	\author[UPV]{Alberto Castillo\corref{corr}}
	\ead{alcasfra@upv.es}
	
	\author[UPV]{Pedro Garcia}
	\ead{pggil@isa.upv.es}
	
	\address[UPV]{Instituto de Autom\'{a}tica e Inform\'{a}tica Industrial (ai2), Universitat Polit\`{e}cnica de Val\`{e}ncia, Valencia, Spain.}
	
	\cortext[corr]{Corresponding author.}
	
	\tnotetext[thanks]{This work has been accepted for publication in Automatica, Volume 57 (2021), Issue 2 (February). It was supported by project FPU15/02008, Ministerio de Educaci\'on y Ciencia, Spain.}

	\begin{abstract}
	Predicting the state of a system in a relatively near future time instant is often needed for control purposes. However, when the system is affected by external disturbances, its future state is dependent on the forthcoming disturbance; which is \textendash in most of the cases\textendash\ unknown and impossible to measure. In this scenario, making predictions of the future system-state is not straightforward and \textendash indeed\textendash\ there are scarce  contributions provided to this issue. This paper treats the following problem: given a LTI system affected by continuously differentiable unknown disturbances, how its future state can be predicted in a sufficiently small time-horizon by high-order observers. Observer design methodologies in order to reduce the prediction errors are given. Comparisons with other solutions are also established.
	\end{abstract}
	

	\end{frontmatter}

	\section{Introduction}
	This note considers the problem of predicting the future-state of the following LTI system:
	\begin{equation}\label{eq:System}
	\begin{aligned}
		&\dot{x}(t) = Ax(t)+B_uu(t-h)+B_{\omega}\omega(t),\\
		&y(t) = C_xx(t)+D_{\omega}\omega(t);
	\end{aligned}
	\end{equation}
	by means of its standard solution:
	\begin{multline}\label{eq:SystemSol}
	x(t+T) = \Phi(T)x(t) + \int_{t}^{t+T}e^{A(t+T-s)}B_uu(s-h)ds +\\ +\int_{t}^{t+T}e^{A(t+T-s)}B_{\omega}\omega(s)ds;
	\end{multline}
	where $t$ represents the current time; $T\geq 0$ the prediction-horizon; $h\geq 0$ a possible input-delay; $x(t)\in\mathbb{R}^n$ the system-state; ${u(t-h)\in\mathbb{R}^m}$ the control input; ${\omega(t)\in\mathbb{R}^q}$ an unknown disturbance-input; $y(t)\in\mathbb{R}^p$ the measurable output; ${A}$, ${B_u}$, $B_{\omega}$, $C_x$, $D_{\omega}$ the nominal system matrices of appropriate dimensions; and $\Phi(T)\triangleq e^{AT}$.

	Predicting the value of $x(t+T)$ is often needed for control purposes. For example, state-predictions are extensively used for closed-loop input-delay compensation (\cite{manitius1979finite,artstein1982linear,krstic2008lyapunov,krstic2010compensation,mazenc2012lyapunov,li2014robustness,cao2017coordinated}); for controlling systems with delays and disturbances (\cite{di2005disturbance, sanz2016predictor, furtat2017disturbance, santos2018enhanced}); for model predictive control~(\cite{mayne2014model,binder2019improved}); and, also, in other applications where estimates of~$x(t+T)$ may be useful, such as collision avoidance~(\cite{polychronopoulos2007sensor}).
	
 	A remarkable problem with disturbed systems is that $\omega(s)$ should be known for ${s\in[t,t+T]}$ in order to implement~\eqref{eq:SystemSol}. This is not the case in almost all practical scenarios, where the future-disturbances are used to be unknown and impossible to measure. By this reason, most part of the previous works simply neglect the disturbance-effect when predicting $x(t+T)$; but, as pointed out by~\cite{sanz2016enhanced}, neglecting the disturbance-effect may lead to significant errors in the computed predictions.
	
	In an effort to reduce the prediction-errors caused by the disturbances, some works have given solutions in order to approximate the second integral in~\eqref{eq:SystemSol}. \cite{lechappe2015new} propose a state-predictor that substitutes this integral by an error-term that compares the actual state with the state-prediction made at $t-T$. Similarly, \cite{sanz2016enhanced} propose a state-predictor that approximates this integral by: $\int_{t-T}^{t}e^{A(t-s)}B_{\omega}\bar{\omega}(s)ds$; being $\bar{\omega}(t)$ an estimate of~${\omega(t+T)}$ that is computed with a disturbance observer, a tracking-differentiator and a Taylor expansion. Both methods have been proved to give better predictions than simply neglecting the disturbance-effect.
	
	This paper shows that, if the disturbance is ${(r+1)}$-times,~${r\in\mathbb{N}\cup\{0\}}$, continuously differentiable, then Eq.~\eqref{eq:SystemSol} can be \textit{directly} approximated by conventional high-order LTI observers. This result provides a novel observer-based state-predictor that is easily implementable and which yields to better accuracy than the solutions of~\cite{lechappe2015new} and \cite{sanz2016enhanced}.
	
	The rest of the paper is organized as follows: Section~\ref{sec:ProblemForm} contains some preliminaries. In Section~\ref{sec:SolutionInget}, the integrals in~\eqref{eq:SystemSol} are rewritten so that $x(t+T)$ is expressed in terms of the \textit{actual} system and disturbance states \textendash where by disturbance state is understood its current value together with its first $r$-derivatives. In Section~\ref{sec:LTIobs}, a high-order LTI observer whose output is an almost direct estimate of~$x(t+T)$ is proposed. Section~\ref{sec:ObsDesign} contains LMI-based design methodologies in order to reduce the prediction errors. Section~\ref{sec:Example} establishes comparisons with respect to the methods of \cite{lechappe2015new} and \cite{sanz2016enhanced} and, finally, Section~\ref{sec:Conclusion} highlights the main conclusions.

	\section{Preliminaries}\label{sec:ProblemForm}
	Let us consider the following assumptions:
	\begin{assum}\label{ass:smooth} The disturbance, $\omega(t)$, is $(r+1)$-times continuously differentiable with $\norm{\omega^{(r+1)}(t)}\leq\epsilon_{r}$; $\forall t\in\mathbb{R}$.
	\end{assum}
	\begin{assum}\label{ass:ZOH}
		The control action, $u(\phi)$, is piece-wise constant and it is defined for all $\phi\in[t-h,\,t+T-h]$.
	\end{assum}
	
	Asm.~\ref{ass:smooth} considers disturbances with a certain degree of differentiability. Smooth disturbances, polynomial-type perturbations, or unknown system-inputs that are generated by exogenous dynamical-systems may satisfy it. Discontinuous or statistical-type perturbations are excluded as they are not differentiable at some \textendash maybe none\textendash \ points.

	Asm.~\ref{ass:ZOH} implies that the control action is generated discreetly and it is introduced to the system via a Zero-Order-Hold~(ZOH) \textendash which is actually the case in almost all computer-based control applications~(\cite{aastrom2013computer}). Also, it states that all its jumping instants, ${t_k\in[t-h,t+T-h]}$, ${k\in \mathbb{N}}$, and its associated control values,~$u_k\in\mathbb{R}^m$, are well-defined and, for prediction purposes, they can be considered as known. This is summarized in the following definition:
	\begin{defi}\label{def:u}
		Let ${t_k\in\mathbb{R}}$, $k\in \{1,...,N\}$, ${N\in\mathbb{N}}$; be the set of time-instants, with ${t_1>t-h}$, ${t_k< t_{k+1}}$, ${t_N<t+T-h}$; such that:
		\begin{equation*}
		\begin{aligned}
			&u(\phi) = u_0\in\mathbb{R}^m,\;\forall \phi\in[t-h,t_{1}),\\
			&u(\phi) = u_k\in\mathbb{R}^m,\;\forall \phi\in[t_k,t_{k+1}),\; 1\leq k\leq N-1,\\
			&u(\phi) = u_{N}\in\mathbb{R}^m,\;\forall \phi\in[t_N,t+T-h].
		\end{aligned}
		\end{equation*}
	\end{defi}

	\section{A solution to the integrals in~\eqref{eq:SystemSol}}\label{sec:SolutionInget}
	Under Asm.~\ref{ass:ZOH} and Def.~\ref{def:u}, the interval $[t,\,t+T]$ can be divided into $N+1$ subintervals where the delayed control action remains constant. Thus, the first integral equals to:
	\begin{multline}\label{eq:IntegU}
	\int_{t}^{t+T}e^{A(t+T-s)}B_uu(s-h)ds = \Gamma_u(t,t_{1}+h)u_{0}+ \\
	+\sum_{k=1}^{N-1}\Gamma_u(t_k+h,t_{k+1}+h)u_k +\Gamma_u(t_{N}+h,t+T)u_{N};
	\end{multline}
	where, for any $\alpha\leq \beta\leq t+T$,
	\begin{equation*}
	\begin{aligned}
		\Gamma_u(\alpha,\beta)&\triangleq \int_{\alpha}^{\beta}e^{A(t+T-s)}B_uds\\
		&=\sum_{j=1}^{\infty}\left[\frac{(t+T-\alpha)^j-(t+T-\beta)^j}{j!}A^{j-1}\right]B_u.
	\end{aligned}
	\end{equation*}
	
	On the other hand, the second integral cannot be directly computed as $\omega(s)$ is not known for $s\in[t,\,t+T]$. However, Ass.~\ref{ass:smooth} allows to rewrite it as follows:
	\begin{lem}
	Under Asm.~\ref{ass:smooth}:
	\begin{multline}\label{eq:IntegW}
	\int_{t}^{t+T}e^{A(t+T-s)}B_{\omega}\,\omega(s)ds = \\ \sum_{i=0}^{r}\left[\frac{T^{i+1}}{(i+1)!}\Gamma_{\omega,i}(T)B_{\omega}\,{\omega}^{(i)}(t)\right] +\mathcal{O}_{r}(t);
	\end{multline}
	with $\Gamma_{\omega,i}(T) \triangleq \sum_{j=0}^{\infty}\frac{(i+1)!}{(i+1+j)!}A^jT^j$; and
	\begin{align}
	\small
	&\mathcal{O}_{r}(t)\triangleq \int_{t}^{t+T}\frac{(t+T-s)^{r+1}}{(r+1)!}\Gamma_{\omega,r}(t+T-s)B_{\omega}\,{\omega}^{(r+1)}(s)ds.\label{eq:Op}
	\end{align}
	\end{lem}
	\begin{proof}
		Refer to \ref{appendix}.
	\end{proof}
	
	The equalities~\eqref{eq:IntegU}-\eqref{eq:IntegW} enable to rewrite~\eqref{eq:SystemSol} as:
	\begin{multline}\label{eq:systemSol2}
		x(t+T) = \Phi(T)x(t) + 	\boldsymbol{\Gamma_{\omega}}(T)\boldsymbol{\omega_r}(t) + \\	\boldsymbol{\Gamma_u}(t,t_1,\,\hdots, t_N,t+T)\boldsymbol{u_{0-N}}+\mathcal{O}_{r}(t),
	\end{multline}
	where $\boldsymbol{\omega_r}(t)\triangleq\; [\omega^T(t),\;\dot{\omega}^T(t),\,\hdots\,, \omega^{(r)\,T}(t)]^T$, $\boldsymbol{u_{0-N}}\triangleq\; [u_0^T,\;u_1^T,\,\hdots\,, u_N^T]^T$ and
	\begin{equation*}
	\begin{aligned}
		\boldsymbol{\Gamma_{\omega}}(T)\triangleq&\; [T\,\Gamma_{\omega,0}(T)B_{\omega},\;\hdots\;,\; \frac{T^{r+1}}{(r+1)!}\Gamma_{\omega,r}(T)B_{\omega}],\\
		\boldsymbol{\Gamma_{u}}(\cdot)\triangleq&\; [\Gamma_u(t,t_1+h),\;\Gamma_u(t_1+h,t_2+h),\, \hdots\\
		&\quad \hdots\;,\Gamma_u(t_{N-1}+h,t_N+h),\Gamma_u(t_N+h,t+T)].
	\end{aligned}
	\end{equation*}
	
	This proves the following result, which permits to directly approximate~\eqref{eq:SystemSol} by high-order observers:
	\begin{thm}\label{thm:Prediction}
		Under Assumptions~\ref{ass:smooth}-\ref{ass:ZOH}. Let $t+T$, $T\geq 0$, be a future time-instant. Consider the set of control inputs of Definition~\ref{def:u} that have been, or could be\footnote{Note that, if $T>h$, then Def.~\ref{def:u} includes a set of future control actions that could be introduced to the system.}, introduced to the system during the time-interval ${[t-h,\,t+T-h]}$. Then, for any given observations, $\hat{x}(t)$, $\boldsymbol{\hat{\omega}_r}(t)$, of $x(t)$, $\boldsymbol{{\omega}_r}(t)$, respectively; the future-state,~${x(t+T)}$, can be predicted by:
		\begin{multline}\label{eq:predEstimate}
		\hat{x}(t+T) = \Phi(T)\hat{x}(t) + 	\boldsymbol{\Gamma_{\omega}}(T)\boldsymbol{\hat{\omega}_r}(t) \\ +\boldsymbol{\Gamma_u}(t,t_1,\hdots,t_N,t+T)\boldsymbol{u_{0-N}},
		\end{multline}
		with an error equal to
		\begin{equation}\label{eq:predError}
		x(t+T)-\hat{x}(t+T) = \mathcal{O}_e(t) + \mathcal{O}_{r}(t);
		\end{equation}
		where $\mathcal{O}_e(t)\triangleq \Phi(T)[x(t)-\hat{x}(t)] + \boldsymbol{\Gamma_{\omega}}(T)[\boldsymbol{{\omega}_r}(t)-\boldsymbol{\hat{\omega}_r}(t)]$.
	\end{thm}
	\begin{proof}
		The proof follows by subtracting~\eqref{eq:systemSol2} and \eqref{eq:predEstimate}.
	\end{proof}
	
	The error equation~\eqref{eq:predError} indicates that the prediction error is caused by two terms: \textbf{i)}~$\mathcal{O}_e(t)$,~which only depends on the current observation errors, $x(t)-\hat{x}(t)$, $\boldsymbol{{\omega}_r}(t)-\boldsymbol{\hat{\omega}_r}(t)$; and \textbf{ii)} $\mathcal{O}_{r}(t)$,~which, under perfect observation conditions (i.e. $\mathcal{O}_e=0$), represents the error caused by predicting~$x(t+T)$ just with the \textit{actual} disturbance state,~$\boldsymbol{{\omega}_r}(t)$, instead of computing it with the whole disturbance \textit{function},~${\omega:[t,t+T]\to \mathbb{R}^q}$. This last term can be nicely seen as the unavoidable error that appears as a consequence of not knowing the future disturbance behavior.
	
	By taking norms in~\eqref{eq:Op}, $\mathcal{O}_{r}(t)$ can be bounded by:
	\begin{equation}\label{eq:boundOp}
		\norm{\mathcal{O}_{r}(t)} \leq  \frac{T^{r+2}}{(r+2)!}\,\mu\,\epsilon_{r},
	\end{equation}
	where $\mu\triangleq \max_{s\in[t,t+T]}\{\norm{\Gamma_{\omega,r}(t+T-s)B_{\omega}}\}$.
	
	It is seen that the size of $\mathcal{O}_{r}(t)$ mainly depends on the prediction horizon and on the $(r+1)$ disturbance-derivative bound \textendash defined in Asm.~\ref{ass:smooth}. On the other hand, $\mathcal{O}_e(t)$ just depends on how the estimates $\hat{ x}(t)$, $\boldsymbol{\hat{\omega}_r}(t)$ are computed.
	
	The next section considers a high-order LTI observer that almost directly implements~\eqref{eq:predEstimate}. In addition, a design methodology in order to quantify and reduce the term $\mathcal{O}_e(t)$ is also included.
	
	\section{A high-order observer for predicting $x(t+T)$}\label{sec:LTIobs}
	Consider the following observer:
	\begin{align}
		&\dot{\hat{\eta}}(t) = (\bar{A}-L\bar{C})\hat{\eta}(t) + \bar{B}_u u(t-h) + Ly(t),\label{eq:observer1}\\
		&\xi(t) = K(T)\hat{\eta}(t) \label{eq:observer2},
	\end{align}
	being $\hat{\eta}(t)\in\mathbb{R}^{l}$, ${l=n+(r+1)q}$,  
	\begin{equation*}
	\begin{aligned}
	&\bar{A}\triangleq \begin{bmatrix}
	A & B_{\omega}\Pi \\
	0 & \Psi 
	\end{bmatrix},\; 
	\Psi\triangleq\begin{bmatrix}
	0_{rq\times q} & I_{rq}\\
	0_{q\times q} & 0_{q\times rq}
	\end{bmatrix},\; \Pi\triangleq \begin{bmatrix}
	I_q & 0_{q\times rq}
	\end{bmatrix} \\
	&\bar{B}_u\triangleq \begin{bmatrix}
	B_u\\0_{(r+1)q\times m}
	\end{bmatrix},\quad \bar{C}\triangleq \begin{bmatrix}
	C_x & D_{\omega} & 0_{p\times rq}
	\end{bmatrix},\\
	&K(T)\triangleq \begin{bmatrix}
	\Phi(T) & \boldsymbol{\Gamma_{\omega}}(T)
	\end{bmatrix}.
	\end{aligned}
	\end{equation*}
	and $L\in\mathbb{R}^{l\times p}$ a free-design matrix.
	
	It can be checked that~\eqref{eq:observer1} is an Luenberger-type observer for the augmented variable ${\eta(t)\triangleq [x^T(t),\,\boldsymbol{\omega_r}^T(t)]^T}$; which, under Asm.~\ref{ass:smooth}, satisfies:
	\begin{equation}\label{eq:augmDyn}
		\dot{\eta}(t)=\bar{A}\eta(t) + \bar{B}_uu(t-h) + \bar{B}_{\omega}\omega^{(r+1)}(t),
	\end{equation}
	being $\bar{B}_{\omega}\triangleq [0_{q\times n+rq},\,I_q]^T$.
	
	Thus, the state-prediction~\eqref{eq:predEstimate} can be directly expressed in terms of $\xi(t)$:
	\begin{equation}\label{eq:statePred}
		\hat{x}(t+T) = \xi(t) + \boldsymbol{\Gamma_u}(t,t_1,\hdots,t_N,t+T)\boldsymbol{u_{0-N}};
	\end{equation}
	and, by~\eqref{eq:observer1} and \eqref{eq:augmDyn}, the term $\mathcal{O}_e(t)$ is given by:
	\begin{equation}\label{eq:Oe}
	\begin{aligned}
		\dot{e}_{\eta}(t)&=(\bar{A}-L\bar{C}){e}_{\eta}(t) + \bar{B}_{\omega}\omega^{(r+1)}(t),\\
		\mathcal{O}_{e}(t)& = K(T)e_{\eta}(t),
	\end{aligned}
	\end{equation}
	being ${e_{\eta}\triangleq \eta(t)-\hat{\eta}(t)}$.

	The observer~\eqref{eq:observer1}-\eqref{eq:observer2} provides two advantages. The first one is that its output is an almost direct estimate of~${x(t+T)}$. It should be just corrected by $\boldsymbol{\Gamma_u}(\cdot)\boldsymbol{u_{0-N}}$ in order to quantify for the effect of the control action.

	The second one is that the knowledge of the dynamics~\eqref{eq:Oe} can be exploited for design purposes, that is: design $L$ so that $\mathcal{O}_e(t)$ is minimized somehow. In fact, a careless design of $L$ may lead to undesirably large prediction errors. By this reason, it is interesting to design it in order to minimize $\mathcal{O}_{e}(t)$.
	
	The next section proposes a LMI-based design methodology to assure that~$\norm{\mathcal{O}_{e}(t)}$ is ultimately bounded by $\gamma\epsilon_{r}$, being $\gamma$ a positive constant to be minimized.
	
	\subsection{An observer design methodology}\label{sec:ObsDesign}
	For the observer design, let us consider that:
	\begin{assum}\label{ass:Observ}
		$(\bar{A},\bar{C})$ is observable.
	\end{assum}
	
	\begin{table*}[b]
		\centering
		\begin{tabular}{|c|cccc|}
			\hline
			& \begin{tabular}{c} observer\\order\end{tabular} & \begin{tabular}{c} observed\\variable\end{tabular} & \begin{tabular}{c}
				computed prediction
			\end{tabular} & \begin{tabular}{c} Type of \\ disturbances  \end{tabular} \\ \hline
			
			\begin{tabular}{c} L\'echapp\'e \\ et al. \end{tabular} &  - & - &  $\hat{ x}(t+T) = x_p(t) + x(t)-x_p(t-T)$	  & \begin{tabular}{c}
				Locally \\ integrable
			\end{tabular} \\
			
			\begin{tabular}{c} Sanz  et al. \end{tabular}	 &  $(r+1)q$  & $\bar{\omega}(t)\approx \omega(t+T)$ &  $\hat{ x}(t+T) = x_p(t)+\int_{t-T}^{t}e^{A(t-s)}B_{\omega}\bar{\omega}(s)ds$  & Asm.~\ref{ass:smooth} \\ 
			
			\begin{tabular}{c} This \\ approach \end{tabular}		& $n+(r+1)q$ & \begin{tabular}{c} $\xi(t)\approx\Phi(T)x(t)+$\\$\int_{t}^{t+T}e^{A(t+T-s)}B_{\omega}\omega(s)ds$\end{tabular} & $\hat{ x}(t+T)=\xi(t)+\Gamma_{u}(\cdot)\boldsymbol{u_{0-N}}$ & Asm.~\ref{ass:smooth} \\
			\hline
		\end{tabular}
		\caption{Comparison of different predictive strategies, where $x_p(t)\triangleq \Phi(T)x(t)+\int_{t}^{t+T}e^{A(t+T-s)}B_uu(s-h)ds$.}
		\label{tab:CualitativeComp}
	\end{table*}

	The observer design procedure is summarized as:
	\begin{lem}\label{lem:designProc}
		Consider a bounded set, $\mathcal{D}$, of the complex plane defined by ${\mathcal{D}\triangleq\left\lbrace z\in\mathbb{C}\,\big\vert\, N+zM+z^*M^T\prec 0  \right\rbrace}$, where ${N=N^T}$ and $M$ are real matrices\footnote{\label{foot:D} Some examples of sets ${\mathcal{D}}$ are given in \cite{chilali1999robust}.} and $z^*$ denotes the conjugate of $z$. Under Asm.~\ref{ass:smooth}-\ref{ass:Observ}, let there exist positive constants $\alpha>0$, $\delta>0$, a symmetric and positive definite matrix ${P\in\mathbb{R}^{l\times l}}$ and a matrix ${Y\in\mathbb{R}^{l\times p}}$ such that
		\begin{subequations}\label{eq:ISSOptSinf}
			\begin{align}
			&\begin{bmatrix}
			P\bar{A} + \bar{A}^TP - Y\bar{C} - \bar{C}^TY^T + 2\delta P & P\bar{B}_{\omega} \\
			\bar{B}_{\omega}^TP 	& -I_q
			\end{bmatrix} \preceq 0, \label{eq:ISSOptSinfLMI}\\
			&P-\alpha K^T(T)K(T)\succeq 0, \label{eq:ISSOptSinfCond1}\\
			&N\otimes P + M\otimes (P\bar{A}-Y\bar{C}) + M^T\otimes (P\bar{A}-Y\bar{C})^T\prec 0 \label{eq:ISSOptSinfCond2}.
			\end{align}
		\end{subequations}
		
		Then, if ${L=P^{-1}Y}$, the error term $\mathcal{O}_{e}(t)$, in~\eqref{eq:Oe}, exponentially approaches \textendash with decay-rate $\delta$\textendash \ to the ball
		\begin{equation}\label{eq:boundOe}
			\norm{\mathcal{O}_{e}}\leq  \epsilon_{r}\sqrt{\frac{1}{2\alpha\delta}},
		\end{equation}
		for any initial state, $e_{\eta}(t_0)$. Moreover, $\text{eig}(\bar{A}-L\bar{C})\in\mathcal{D}$.
	\end{lem}
	\begin{proof}
		Define ${V(t)\triangleq e_{\eta}(t)^T Pe_{\eta}(t)}$ as a Lyapunov function. Under Asm.~\ref{ass:smooth}, if there exist $\delta>0$ such that:
		\begin{equation}\label{eq:ISS}
			\dot{{V}}(t)+2\delta V(t) - \norm{\omega^{(r+1)}(t)}^2 \leq 0,
		\end{equation}
		then, for any $e_{\eta}(t_0)$, $e_{\eta}(t)$ exponentially approaches to the ellipsoid ${\mathcal{E}\triangleq \left\lbrace e_{\eta}\in\mathbb{R}^{l}\, \big\vert \,e_{\eta}^TPe_{\eta}\leq \frac{\epsilon_{r}^2}{2\delta}\right\rbrace}$ (\cite{fridman2014introduction}).
		
		By differentiating $V(t)$ and substituting~\eqref{eq:Oe}, it is proved that~\eqref{eq:ISS} holds if the LMI~\eqref{eq:ISSOptSinfLMI} is satisfied. 
		
		Eq.~\eqref{eq:ISSOptSinfCond1} guarantees that $\alpha\norm{\mathcal{O}_{e}(t)}^2 = \alpha e_{\eta}(t)^TK^TKe_{\eta}(t)\leq e_{\eta}^T(t)Pe_{\eta}(t)$. Therefore, it assures that ${\norm{\mathcal{O}_{e}}^2\leq \frac{\epsilon_{r}^2}{2\alpha\delta}}$ for any~$e_{\eta}$ belonging to~$\mathcal{E}$.
		
		Finally, as shown in \cite{chilali1999robust}, condition~\eqref{eq:ISSOptSinfCond2} guarantees that ${\text{eig}(\bar{A}-L\bar{C}_2)\in\mathcal{D}}$.
	\end{proof}
	
	\begin{rmk}\label{rmk:minOe}
		In order to minimize~\eqref{eq:boundOe}, Lemma~\ref{lem:designProc} can be optimized in order to maximize $\alpha\delta$. This can be done by standard LMI solvers. Condition~\eqref{eq:ISSOptSinfCond2}, which forces that ${\text{eig}(\bar{A}-L\bar{C}_2)\in\mathcal{D}}$ \textendash with $\mathcal{D}$ bounded\textendash, avoids the numerical divergence~$\norm{L}\to\infty$ that may appear when maximizing~$\alpha\delta$. This constraint also permits to include in the design process additional restrictions that are often required in practice, such as: maximum observer bandwidth or maximum overshoot.
	\end{rmk}

	At this point, the following result has been proved.
	\begin{thm}\label{thm:Prediction2}
		Under Assumptions~\ref{ass:smooth}-\ref{ass:Observ}. Let $t+T$, $T\geq 0$, be a future time-instant. Consider the observer~\eqref{eq:observer1}-\eqref{eq:observer2}, with~$L$ satisfying the conditions of Lemma~\ref{lem:designProc}, and the set of control inputs of Definition~\ref{def:u} that have been, or could be, introduced to the system during ${[t-h,\,t+T-h]}$. Then, the future state, $x(t+T)$, can be predicted by~\eqref{eq:statePred} with an error equal to: $x(t+T)-\hat{ x}(t+T)=\mathcal{O}_{r}(t)+\mathcal{O}_e(t)$; where~$\mathcal{O}_{r}(t)$ satisfies the bound~\eqref{eq:boundOp}, while $\mathcal{O}_e(t)$ is ultimately bounded by~\eqref{eq:boundOe}; being both null for $\epsilon_{r}=0$.
	\end{thm}
	
	\section{Comparisons}\label{sec:Example}
	This section compares the solution of this paper with respect to \cite{lechappe2015new} and \cite{sanz2016enhanced}.
		
	Table~\ref{tab:CualitativeComp} contains a brief summary of each method. The solution provided by \cite{lechappe2015new} is the most easily implementable as it does not need to implement any observer. As shown in Table~\ref{tab:CualitativeComp}, this method is based on replacing the disturbance-dependent integral in~\eqref{eq:SystemSol} by an error-term that compares the actual state with the prediction,  $x_p(t)\triangleq \Phi(T)x(t)+\int_{t}^{t+T}e^{A(t+T-s)}B_uu(s-h)ds$, made at $(t-T)$. In this way, null prediction error is assured for constant disturbances. On the other hand, the predictor in \cite{sanz2016predictor} \textendash originally developed for matched disturbances\textendash \ implements an observer of order ${(r+1)q}$ in order to predict $\bar{\omega}(t)\approx{\omega(t+T)}$. This prediction is then used to approximate the integral by: $\int_{t-T}^{t}e^{A(t-s)}B_{\omega}\bar{\omega}(s)ds$; which can be numerically computed as it just depends on the past values of $\bar{\omega}(t)$.
		
	In contrast, the solution of this paper constructs an observer of order $n+(r+1)q$ which \textit{directly} estimates $\xi(t)\approx\Phi(T)x(t)+\int_{t}^{t+T}e^{A(t+T-s)}B_{\omega}\omega(s)ds$. This leads to significant advantages in terms of implementation simplicity and numerical accuracy; as shown in the next section.
	\begin{table}[t]
		\begin{center}
			\begin{tabular}{|c|cccc|}
				$\Omega$ & This approach & L\'echapp\'e et al. & Sanz et al. &\\
				\hline
				$0.5$ 	& 1.7e-4 		& 0.60  	& 8.9e-4  & \\
				$2$ 	& 0.043 		& 2.23 	& 0.21 &\\
				$4$ 	& 0.66 		& 3.47 	& 2.82 & \\
				\hline
			\end{tabular}
			\caption{$\max_{t+T\in[T,\infty]}\{\norm{x(t+T)-\hat{ x}(t+T)}\}$.}
			\label{tab:BoundsComparisonOmega}
		\end{center}
	\end{table}
	\begin{figure}[t]
		\centering
		\includegraphics[width=3.3in]{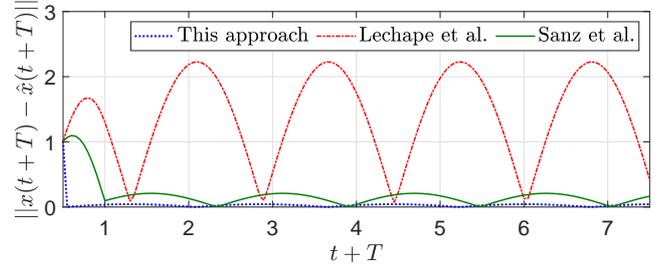}
		\caption{Prediction errors for ${\Omega=2}$~rad/s.}
		\label{fig:SimResult}
	\end{figure}
	
	\subsection{Numerical simulation}
	For a numerical comparison, let us evaluate the proposed observer-based predictor with the same example considered by~\cite{lechappe2015new,sanz2016enhanced}:
	\begin{equation}\label{eq:SystemLech}
	\begin{aligned}
	&\dot{x}(t)=\begin{bmatrix}
	0 & 1\\-9 & 3\end{bmatrix}x(t) + \begin{bmatrix}
	0\\1
	\end{bmatrix}u(t-h) + \begin{bmatrix}
	0\\1
	\end{bmatrix}\omega(t),
	\end{aligned}
	\end{equation}
	with $\omega(t)=3\sin(\Omega t)$; $\Omega\geq 0$~(rad/s); and $x(0)=0$.
	
	An accurate state-prediction is required in order to approximately implement a stabilizing feedback with Finite Spectrum Assignment~(FSA):  $u(t)=-[45,\,18]\hat{x}(t+T)$, with $T=h$. To satisfy Asm.~\ref{ass:ZOH}, the control action is introduced to the system via a ZOH with period $0.05$~sec.
	
	Table~\ref{tab:BoundsComparisonOmega} and Fig.~\ref{fig:SimResult} contain the resulting prediction errors with each method for different disturbance frequencies. It can be seen that the solution of this paper leads to significant improvements in terms of accuracy. 
	
	In the simulations, the proposed solution has been implemented according to Thm.~\ref{thm:Prediction2}, starting from $\hat{\eta}(0)=\eta(0)$, with $L=$[10.28, 4.32, 512, 9680, 1.16e05, 7.96\text{e}05;
	0.22, 58.8, 2960, 83900, 1.5\text{e}06, 1.38\text{e}07]	resulting from the optimization of Lem.~\ref{lem:designProc}, with $r=3$, $\mathcal{D}$ a disk centered at $(0,0)$ with radius $40$ and $\delta\alpha = 6.87\text{e}05$. For a fair comparison, the integral-term in $x_p(t)$ has also been computed by the ZOH method with sample period 0.05 sec. In addition, the variable ${\bar{\omega}(t)\approx \omega(t+T)}$ of \cite{sanz2016enhanced} has been computed by Taylor with the same estimate $\boldsymbol{\hat{\omega}_r}(t)$ given by~\eqref{eq:observer1}; while $\int_{t-T}^{t}e^{A(t-s)}B_{\omega}\bar{\omega}(s)ds$ has been numerically solved by the trapezoidal method with 100 intervals.
	
	It should be noted that none disturbance information has been used in order to compute the predictions. But, if the bound $\epsilon_{r}$ in Asm.~\ref{ass:smooth} was known, then Eqs.~\eqref{eq:boundOp}, \eqref{eq:boundOe} could be employed to get a numerical upper-bound for the prediction error. Here, $\epsilon_{r}=3\Omega^4$. Thus, the error $x(t+T)-\hat{ x}(t+T)$ is ultimately bounded by $0.00356\Omega^4$; which, if it is solved for $\Omega=\{0.5,\,2,\,4\}$ rad/s, it can be checked that it is quite close to the resulting errors of Table~\ref{tab:BoundsComparisonOmega}.

	\section{Conclusions}\label{sec:Conclusion}
	This paper has developed a novel observer-based methodology for future-state prediction in disturbed systems. The given solution is easily implementable and it yields reduced errors if compared with other approaches.
	
	The prediction-error of the given solution is proportional to the parameter $\epsilon_{r}$ of Asm.~\ref{ass:smooth}. Thus, reduced errors will be attained for disturbances with small $\epsilon_{r}$, achieving null-error for $\epsilon_{r}=0$. In addition \textendash and independently of the explicit value of $\epsilon_{r}$\textendash \ an observer design methodology has been provided in order to minimize the prediction-errors.

	The global requirement of Asm.~\ref{ass:smooth} \textendash i.e. $\forall\,t\in\mathbb{R}$\textendash \ could be relaxed just to a finite time-interval: $t\in[t_0,t_f]$, being ${t_f-t_0>T}$ and $\epsilon_{r}$ local bound. In this case, the results of this paper would still being valid for $t\in[t_0,t_f-T]$. Finally, Theorem~\ref{thm:Prediction} and Eqs.~\eqref{eq:observer1}-\eqref{eq:observer2}, \eqref{eq:statePred} remain valid for time-varying prediction-horizons, i.e. $T\triangleq T(t)$.

	\appendix
	\section{Proof of equality~\eqref{eq:IntegW}}\label{appendix}
	Let $v_{i}\triangleq \int_{t}^{t+T}A^{-i}e^{A(t+T-s)}B_{\omega}\omega^{(i)}(s)ds$. Under Ass.~\ref{ass:smooth}, $v_i$ can be integrated by parts for any $i<r+1$; leading to:
	\begin{equation}\label{eq:recInteg}
		v_i  =-A^{-(i+1)}B_{\omega}\omega^{(i)}(t+T)+A^{-(i+1)}e^{AT}B_{\omega}\omega^{(i)}(t) +v_{i+1}
	\end{equation}

	The recursive Eq.~\eqref{eq:recInteg} allows to express $v_0$ as:
	\begin{equation}\label{eq:IntegDem}
		\begin{aligned}
			v_0 =&-A^{-1}B_{\omega}\omega(t+T) -  ... - A^{-(r+1)}B_{\omega}\omega^{(r)}(t+T) \\
			&+ A^{-1}e^{AT}B_{\omega}\omega(t) + ... + A^{-(r+1)}e^{AT}B_{\omega}\omega^{(r)}(t) \\
			&+ v_{r+1}.
		\end{aligned}
	\end{equation}
	
	Let $h_i\triangleq\int_{t}^{t+T}\frac{(t+T-s)^i}{i!}\omega^{(r+1)}(s)ds$.  By Taylor, it holds that:
	\begin{equation}\label{eq:Taylor}
		\begin{aligned}
		&\omega(t+T) = \omega(t) + T\dot{\omega}(t) + ... + \frac{T^r}{r!}\omega^{(r)}(t) + h_r\\
		&\dot{\omega}(t+T) = \dot{\omega}(t) + T\ddot{\omega}(t) + ... + \frac{T^{r-1}}{(r-1)!}\omega^{(r)}(t) + h_{r-1} \\
		&\hspace{0.5cm}\vdots\\
		&\omega^{(r)}(t+T) = \omega^{(r)}(t) + h_0
		\end{aligned}
	\end{equation}
	
	Thus, by substituting~\eqref{eq:Taylor} into \eqref{eq:IntegDem}, and after rearranging terms, the integral $v_0$ results in:
	\begin{equation}\label{eq:IntegDem2}
	\begin{aligned}
		v_0 =& \big[-A^{-1}+A^{-1}e^{AT}\big]B_{\omega}\omega(t)\\
		     &+\big[-A^{-2}-TA^{-1}+A^{-2}e^{AT}\big]B_{\omega}\omega^{(1)}(t)\\
		& \quad\vdots\\
		& +\Big[-A^{-(r+1)}-TA^{-r}-\frac{T^2}{2!}A^{-r-1} - ...
		\\& \hspace{1cm}... -\frac{T^r}{r!}A^{-1} + A^{-(r+1)}e^{AT}\Big] B_{\omega}\omega^{(r)}(t)\\
		&+v_{r+1} - A^{-1}B_{\omega}h_r - ... - A^{-(r+1)}B_{\omega}h_0.
	\end{aligned}
	\end{equation}
	
	Equality~\eqref{eq:IntegW} follows by: i) Noting that $v_0$ equals to the integral in the LHS of~\eqref{eq:IntegW}; and ii) expressing $e^{AT}$ as power-series and cancelling the terms that are added and subtracted in each term of the RHS of~\eqref{eq:IntegDem2}.

	\bibliography{biblio}

	\end{document}